\documentclass[dvipsnames]{article}

\usepackage[utf8]{inputenc}

\usepackage{amsmath}
\usepackage{amssymb}

\usepackage{amsthm}

\usepackage{wrapfig}

\usepackage{xspace}

\usepackage{colortbl}

\usepackage{booktabs}

\usepackage{mathtools}

\usepackage{tikz}
\usetikzlibrary{arrows,decorations.text,decorations.pathmorphing,decorations.pathreplacing,positioning,patterns,tikzmark,arrows.meta,backgrounds,shapes.geometric}
\tikzset{  
  >=stealth',
  plainnode/.style 	= {draw, thick,circle, minimum size=8mm, inner sep=0mm}
}

\newcommand{\texorpdfstring}[2]{#2}

\newtheorem{lemma}{Lemma}
\newtheorem{theorem}{Theorem}

\newcommand{\nats}{\mathbb{N}}
\renewcommand{\epsilon}{\varepsilon}
\renewcommand{\phi}{\varphi}

\newcommand{\pow}[1]{2^{#1}}
\newcommand{\cceq}{\mathop{::=}}
\newcommand{\set}[1]{\{#1\}}

\newcommand{\F}{\mathop{\mathbf{F}\vphantom{a}}\nolimits}
\newcommand{\G}{\mathop{\mathbf{G}\vphantom{a}}\nolimits}
\DeclareMathOperator{\U}{\mathbf{U}}
\newcommand{\X}{\mathop{\mathbf{X}\vphantom{a}}\nolimits}

\newcommand{\ltl}{\mathrm{LTL}\xspace}
\newcommand{\qptl}{\mathrm{QPTL}\xspace}
\newcommand{\ctlstar}{\mathrm{CTL^*}\xspace}
\newcommand{\hyltl}{\mathrm{Hyper\-LTL}\xspace}
\newcommand{\hyqptl}{\mathrm{Hyper\-QPTL}\xspace}
\newcommand{\hyqptlplus}{\mathrm{Hyper\-QPTL^+}\xspace}
\newcommand{\sohyltl}{\mathrm{Hyper^2LTL}\xspace}

\newcommand{\hyctlstar}{\mathrm{HyperCTL^*}\xspace}

\newcommand{\var}{\mathcal{V}}
\newcommand{\fovar}{\mathcal{V}_1}
\newcommand{\sovar}{\mathcal{V}_2}
\newcommand{\ap}[0]{\mathrm{AP}}

\newcommand{\univar}{X_a}
\newcommand{\unidisvar}{X_d}

\newcommand{\tower}{\textsc{Tower}\xspace}

\newcommand{\myquot}[1]{``#1''}

\newcommand{\tsys}{\mathfrak{T}}
\newcommand{\traces}{\mathrm{Tr}}

\newcommand{\inprop}{\texttt{x}}

\newcommand{\contcard}{\mathfrak{c}}

\newcommand{\prop}{\texttt{p}}
\newcommand{\argone}{\texttt{arg1}}
\newcommand{\argtwo}{\texttt{arg2}}
\newcommand{\res}{\texttt{res}}
\newcommand{\add}{\texttt{add}}
\newcommand{\mult}{\texttt{mult}}

\newcommand{\plustimes}{{(+,\cdot)}}
\newcommand{\hyperize}{{\mathit{hyp}}}

\newcommand{\pair}{\mathit{pair}}
\newcommand{\istrace}{\mathit{isTrace}}

\newcommand{\alltraces}{\mathit{all}}

\newcommand{\skolemformat}{\mathit{format}}
\newcommand{\skolemcorrect}{\mathit{correct}}

\newcommand{\natsstruct}{(\nats, +, \cdot, <, \in)}

\newcommand{\propvar}{\texttt{q}}

\newcommand{\marker}{\texttt{m}}
\newcommand{\complete}{\mathit{cmplt}}
\newcommand{\merge}{^\smallfrown}

\newcommand{\enc}{\mathit{enc}}

\newcommand{\all}{\mathit{\!all}}

\newcommand{\consistent}{\mathit{cons}}
\newcommand{\replace}[2]{\mathit{repl}_{#1}(#2)}

\newcommand{\temp}{\mathit{\!tmp}}

\usepackage{fullpage}

\title{The Complexity of HyperQPTL}

\author{
Gaëtan Regaud (ENS Rennes, Rennes, France)\\
Martin Zimmermann (Aalborg University, Aalborg, Denmark)}
\date{}

\begin{document}

\maketitle

\begin{abstract}
HyperQPTL and HyperQPTL$^+$ are expressive specification languages for hyperproperties, properties that relate multiple executions of a system.
Tight complexity bounds are known for HyperQPTL finite-state satisfiability and model-checking.

Here, we settle the complexity of satisfiability for HyperQPTL as well as satisfiability, finite-state satisfiability, and model-checking for HyperQPTL$^+$: the former is $\Sigma^2_1$-complete, the latter are all equivalent to truth in third-order arithmetic, i.e., all four are very undecidable.
\end{abstract}

\section{Introduction}
\label{sec:intro}

Hyperproperties~\cite{ClarksonS10} are properties relating multiple executions of a system and have found applications in security and privacy, epistemic reasoning, and verification. 
Temporal logics have been introduced to express hyperproperties, e.g., $\hyltl$ and $\hyctlstar$~\cite{ClarksonFKMRS14} (which extend $\ltl$ and $\ctlstar$ with trace quantification), $\sohyltl$~\cite{DBLP:conf/cav/BeutnerFFM23} (second-order $\hyltl$, which extends $\hyltl$ with quantification over sets of traces), and many more.
Here, we focus on the most important verification problems:
\begin{itemize}
    \item Satisfiability: Given a sentence~$\phi$, is there a nonempty set~$T$ of traces such that $T \models \phi$?
    \item Finite-state satisfiability: Given a sentence~$\phi$, is there a finite transition system~$\tsys$ such that $\tsys \models \phi$?
    \item Model-checking: Given a sentence~$\phi$ and a finite transition system~$\tsys$, do we have $\tsys \models \phi$?
\end{itemize}

This work is part of a research program~\cite{FinkbeinerRS15,Rabe16diss,FinkbeinerH16,hyperltlsat,DBLP:conf/csl/Frenkel025,fragments,rz25} settling the complexity of these problems for hyperlogics. 
Most importantly for applications in verification, the model-checking problems for $\hyltl$ and $\hyctlstar$ are decidable, albeit $\tower$-complete~\cite{FinkbeinerRS15,Rabe16diss,MZ20}. 
However, satisfiability is typically much harder. 
In fact, the satisfiability problems are typically highly undecidable, i.e., we measure their complexity by placing them in the arithmetic or analytic hierarchy, or even beyond:
Intuitively, first-order arithmetic is predicate logic over the signature~$(+, \cdot, <)$ where quantification ranges over natural numbers.
Similarly, second-order arithmetic adds quantification over sets of natural numbers to first-order arithmetic while third-order arithmetic adds quantification over sets of sets of natural numbers to second-order quantification.
Figure~\ref{fig_hierarchies} depicts the arithmetic, analytic, and \myquot{third} hierarchy, each spanned by the classes of languages definable by restricting the number of alternations of the highest-order quantifiers, e.g., $\Sigma_n^0$ contains languages definable by formulas of first-order arithmetic with $n-1$ quantifier alternations, starting with an existential one.

\begin{figure*}[h]
    \centering
    
  \scalebox{1}{
  \begin{tikzpicture}[xscale=1.0,yscale=.6,thick]

    \fill[fill = gray!25, rounded corners] (-1.05,2) rectangle (0.5,-1.5);
    \fill[fill = gray!25, rounded corners] (.6,2) rectangle (14.5,-1.5);

    \node[align=center] (s00) at (0,0) {$\Sigma^0_0$ \\ $=$ \\ $\Pi^0_0$} ;
    \node (s01) at (1,1) {$\Sigma^0_1$} ;
    \node (p01) at (1,-1) {$\Pi^0_1$} ;
    \node (s02) at (2,1) {$\Sigma^0_2$} ;
    \node (p02) at (2,-1) {$\Pi^0_2$} ;
    \node (s03) at (3,1) {$\Sigma^0_3$} ;
    \node (p03) at (3,-1) {$\Pi^0_3$} ;
    \node (s04) at (4,1) {$\cdots$} ;
    \node (p04) at (4,-1) {$\cdots$} ;
    
    \node[align=center] (s10) at (5,0) {$\Sigma^1_0 $ \\ $=$ \\ $ \Pi^1_0$} ;
    \node (s11) at (6,1) {$\Sigma^1_1$} ;
    \node (p11) at (6,-1) {$\Pi^1_1$} ;
    \node (s12) at (7,1) {$\Sigma^1_2$} ;
    \node (p12) at (7,-1) {$\Pi^1_2$} ;
    \node (s13) at (8,1) {$\Sigma^1_3$} ;
    \node (p13) at (8,-1) {$\Pi^1_3$} ;
    \node (s14) at (9,1) {$\cdots$} ;
    \node (p14) at (9,-1) {$\cdots$} ;
    
    \node[align=center] (s20) at (10,0) {$\Sigma^2_0$ \\ $ =$ \\ $ \Pi^2_0$} ;
    \node (s21) at (11,1) {$\Sigma^2_1$} ;
    \node (p21) at (11,-1) {$\Pi^2_1$} ;
    \node (s22) at (12,1) {$\Sigma^2_2$} ;
    \node (p22) at (12,-1) {$\Pi^2_2$} ;
    \node (s23) at (13,1) {$\Sigma^2_3$} ;
    \node (p23) at (13,-1) {$\Pi^2_3$} ;
    \node (s24) at (14,1) {$\cdots$} ;
    \node (p24) at (14,-1) {$\cdots$} ;
    
    \foreach \i in {0,1,2} {
      \draw (s\i0) -- (s\i1) ;
      \draw (s\i0) -- (p\i1) ;
      \draw (s\i1) -- (s\i2) ;
      \draw (s\i1) -- (p\i2) ;
      \draw (p\i1) -- (s\i2) ;
      \draw (p\i1) -- (p\i2) ;
      \draw (s\i2) -- (s\i3) ;
      \draw (s\i2) -- (p\i3) ;
      \draw (p\i2) -- (s\i3) ;
      \draw (p\i2) -- (p\i3) ;
      \draw (s\i3) -- (s\i4) ;
      \draw (s\i3) -- (p\i4) ;
      \draw (p\i3) -- (s\i4) ;
      \draw (p\i3) -- (p\i4) ;
    }
    \foreach \i [evaluate=\i as \iplus using int(\i+1)] in {0,1} {
      \draw (s\i4) -- (s\iplus0) ;
      \draw (p\i4) -- (s\iplus0) ;
}
      \node[] at (-0.25,1.7) {\scriptsize Decidable} ;
      \node at (13.4,1.7) {\scriptsize Undecidable} ;


      \node[align=left,font=\small,
      rounded corners=2pt] at (3.2,1.7)
      (re) {\scriptsize Recursively enumerable} ;
      \draw[-stealth,rounded corners=2pt] (s01) |- (re) ;

    \path (4.25,-1.75) edge[decorate,decoration={brace,amplitude=3pt}]
    node[below] {\begin{minipage}{4cm}\centering
    	arithmetical hierarchy $\equiv$\\ first-order arithmetic 
    \end{minipage}} (.75,-1.75) ;

    \path (9.25,-1.75) edge[decorate,decoration={brace,amplitude=3pt}]
    node[below] {\begin{minipage}{4.4cm}\centering
    	analytical hierarchy $\equiv$\\ second-order arithmetic 
    \end{minipage}} (5,-1.75) ;

    \path (14.25,-1.75) edge[decorate,decoration={brace,amplitude=3pt}]
    node[below] {\begin{minipage}{4cm}\centering
    	\myquot{the third hierarchy}  {$\equiv$}\\ third-order arithmetic 
    \end{minipage}} (10,-1.75) ;

  \end{tikzpicture}}
    
    \caption{The arithmetical hierarchy, the analytical hierarchy, and beyond.}
    \label{fig_hierarchies}
\end{figure*}
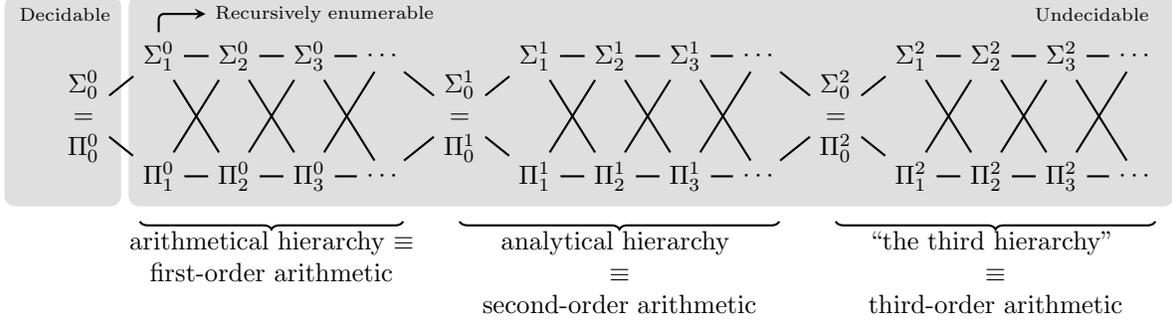

Here, we study the logics~$\hyqptl$ and $\hyqptlplus$ which extend $\hyltl$ by quantification over propositions~\cite{Rabe16diss,FinkbeinerHHT20}, just as $\qptl$ extends $\ltl$ with quantification over propositions~\cite{qptl}. $\qptl$ is, unlike $\ltl$, able to express all $\omega$-regular properties while $\hyqptl$ is, unlike $\hyltl$, able to express, e.g., promptness and epistemic properties~\cite{FinkbeinerHHT20}.

The difference between $\hyqptl$ and $\hyqptlplus$ manifests itself in the semantics of propositional quantification.
Recall that hyperlogics are evaluated over sets~$T$ of traces.
In $\hyqptl$, the quantification of a proposition~$\prop$ reassigns the truth value of $\prop$ in $T$ in a uniform way over all traces, i.e., after quantifying $\prop$ all traces coincide on their truth values for $\prop$.
In $\hyqptlplus$ on the other hand, the quantification of a proposition~$\prop$ reassigns truth values of $\prop$ to traces in $T$ individually.
Said differently, in $\hyqptl$ one quantifies over a single sequence of truth values (i.e., a sequence in $\set{0,1}^\omega$) while in $\hyqptlplus$ one quantifies over a set of sequences of truth values (i.e., a subset of $\set{0,1}^\omega$).
Hence, one expects $\hyqptlplus$ to be more expressive than $\hyqptl$.

And indeed, Finkbeiner et al.\ showed that $\hyqptlplus$ model-checking is undecidable~\cite{FinkbeinerHHT20}, while it is $\tower$-complete for $\hyqptl$~\cite{Rabe16diss}.
It is also known that $\hyqptl$ finite-state satisfiability is $\Sigma^0_1$-complete~\cite{Rabe16diss}, i.e., complete for the recursively-enumerable languages, but the exact complexity of $\hyqptl$ satisfiability is open: it is only known to be $\Sigma^1_1$-hard, as satisfiability is already $\Sigma_1^1$-complete for its fragment $\hyltl$~\cite{FinkbeinerH16,hyperltlsat}.
For $\hyqptlplus$, much less is known: As mentioned above, model-checking $\hyqptlplus$ is undecidable~\cite{FinkbeinerHHT20}, but its exact complexity is open, as is the complexity of satisfiability and finite-state satisfiability.

Here, we settle the complexity of all four open problems by showing that $\hyqptl$ satisfiability is $\Sigma^2_1$-complete (Section~\ref{sec_hyqptl}) while all three problems for $\hyqptlplus$ are equivalent to truth in third-order arithmetic (Section~\ref{sec_hyqptlplus}).
These latter results are obtained by showing that $\hyqptlplus$ and $\sohyltl$ have the same expressiveness.
This confirms the expectation that $\hyqptlplus$ is more expressive than $\hyqptl$: the non-uniform quantification of propositions allows to simulate quantification over arbitrary sets of traces.

\section{Preliminaries}
\label{sec:prels}

We denote the nonnegative integers by $\nats$. 
An alphabet is a nonempty finite set. 
The set of infinite words over an alphabet~$\Sigma$ is denoted by $\Sigma^\omega$.
Let $\ap$ be a nonempty finite set of atomic propositions. 
A trace over $\ap$ is an infinite word over the alphabet~$\pow{\ap}$.
Given a subset~$\ap' \subseteq \ap$, the $\ap'$-projection of a trace~$t(0)t(1)t(2) \cdots$ over $\ap$ is the trace~$(t(0) \cap \ap')(t(1) \cap \ap')(t(2) \cap \ap') \cdots \in (\pow{\ap'})^\omega$.
The $\ap'$-projection of $T \subseteq (\pow{\ap})^\omega$ is defined as the set of $\ap'$-projections of traces in $T$.
We write $t_0 =_{\ap'} t_1$ ($T_0 =_{\ap'} T_1$) if the $\ap'$-projections of $t_0$ and $t_1$ ($T_0$ and $T_1$) are equal.
Now, let $\ap$ and $\ap'$ be two disjoint sets, let $t$ be a trace over~$\ap$, and let $t'$ be a trace over $\ap'$. 
Then, we define $t \merge t'$ as the pointwise union of $t$ and $t'$, i.e., $t \merge t'$ is the trace over $\ap \cup \ap'$ defined as $(t(0) \cup t'(0))(t(1) \cup t'(1))(t(2) \cup t'(2))\cdots$.

A transition system~$\tsys = (V,E,I, \lambda)$ consists of a finite nonempty set~$V$ of vertices, a set~$E \subseteq V \times V$ of (directed) edges, a nonempty set~$I \subseteq V$ of initial vertices, and a labeling~$\lambda\colon V \rightarrow \pow{\ap}$ of the vertices by sets of atomic propositions.
We assume that every vertex has at least one outgoing edge.
A path~$\rho$ through~$\tsys$ is an infinite sequence~$\rho(0)\rho(1)\rho(2)\cdots$ of vertices with $\rho(0) \in I$ and $(\rho(n),\rho(n+1))\in E$ for every $n \ge 0$.
The trace of $\rho$ is defined as $ \lambda(\rho ) = \lambda(\rho(0))\lambda(\rho(1))\lambda(\rho(2))\cdots$.
The set of traces of $\tsys$ is $\traces(\tsys) = \set{\lambda(\rho) \mid \rho \text{ is a path through $\tsys$}}$.
Note that this set is always nonempty, as $I$ is nonempty and as there are no terminal vertices.

To capture the complexity of undecidable problems, we consider formulas of arithmetic, i.e., predicate logic with signature~$(+, \cdot, <, \in)$, evaluated over the structure~$\natsstruct$ (see, e.g.,~\cite{Rogers87}). 
A type~$0$ object is a natural number in $\nats$, a type~$1$ object is a subset of $\nats$, and a type~$2$ object is a set of subsets of $\nats$.

First-order arithmetic allows to quantify over type~$0$ objects, second-order arithmetic allows to quantify over type~$0$ and type~$1$ objects, and third-order arithmetic allows to quantify over type~$0$, type~$1$, and type~$2$ objects.
Note that every fixed natural number is definable in first-order arithmetic, so we freely use them as syntactic sugar. Similarly, equality can be expressed using~$<$, so we use it as syntactic sugar as well.

Truth in third-order arithmetic is the following problem: given a sentence~$\phi$ of third-order arithmetic, does $\natsstruct $ satisfy $\phi$ (written $\natsstruct\models\phi$ as usual)?
Furthermore, $\Sigma^2_1$ contains the sets of the form~$
    \set{n \in\nats \mid \natsstruct \models \phi(n)},
    $
    where $\phi$ is a formula of the form~$\exists \mathcal{Y}_1\ldots \exists \mathcal{Y}_k.\ \psi(x, \mathcal{Y}_1, \ldots, \mathcal{Y}_k,)$ where the $\mathcal{Y}_j$ are third-order variables and $\psi$ is a formula of arithmetic containing only first- and second-order quantifiers.

\section{\texorpdfstring{$\hyqptl$}{HyperQPTL}}
\label{sec_hyqptl}

In this section, we introduce $\hyqptl$ and then settle the complexity of its satisfiability problem.

\subsection{Syntax and Semantics}
\label{subsec_hyperqptl_syse}
Let $\var$ be a countable set of trace variables. 
The formulas of $\hyqptl$~\cite{Rabe16diss,FinkbeinerHHT20} are given by the grammar
\begin{align*}
    \phi &{} \cceq {} \exists \pi.\ \phi \mid \forall \pi.\ \phi \mid \exists \propvar.\ \phi \mid \forall \propvar.\ \phi \mid \psi \\
    \psi&{} \cceq{} \prop_\pi \mid \propvar \mid \lnot \psi \mid \psi \lor \psi \mid \X \psi \mid \F \psi
\end{align*}
where $\prop$ and $\propvar$ range over $\ap$ and where $\pi$ ranges over $\var$.
Note that there are two types of atomic formulas, i.e., propositions labeled by traces on which they are evaluated ($\prop_\pi$ with $\prop \in \ap$ and $\pi \in \var$) and unlabeled propositions ($\propvar \in \ap$). We use different letters for the propositions in these cases, but let us stress again that both $\prop$ and $\propvar$ are propositions in $\ap$.
A formula is a sentence, if every occurrence of an atomic formula~$\prop_\pi$ is in the scope of a quantifier binding~$\pi$ and every occurrence of an atomic formula~$\propvar$ is in the scope of a quantifier binding $\propvar$.
Finally, note that the only temporal operators we have in the syntax are next~($\X$) and eventually~($\F$), as the other temporal operators like always~($\G$) and until~$(\U)$ are syntactic sugar in $\hyqptl$~\cite{kaivolaphd}. 
Hence, we use them freely in the following, just as conjunction, implication, and equivalence.

A trace assignment is a partial mapping~$\Pi\colon \var \rightarrow (\pow{\ap})^\omega$.
Given a variable~$\pi \in \var$, and a trace~$t$, we denote by $\Pi[\pi \mapsto t]$ the trace assignment that coincides with $\Pi$ on all variables but $\pi$, which is mapped to $t$. 
Let $\propvar \in \ap$, let $t \in (\pow{\ap})^\omega$ be a trace over $\ap$, and let $t_\propvar \in (\pow{\set{\propvar}})^\omega$ be a trace over~$\set{\propvar}$.
We define the trace~$t[\propvar\mapsto t_\propvar] = t'{} \merge t_\propvar$, where $t'$ is the $(\ap\setminus\set{\propvar})$-projection of $t$: Intuitively, the occurrences of $\propvar$ in $t$ are replaced according to $t_\propvar$.
We lift this to sets~$T$ of traces by defining $T[\propvar\mapsto t_\propvar] = \set{t[\propvar\mapsto t_\propvar] \mid t \in T}$. 
All traces in $T[\propvar\mapsto t_\propvar]$ have the same $\set{\propvar}$-projection, i.e.,  $t_\propvar$.

Now, for a trace assignment~$\Pi$, a position~$i \in\nats$, and a set~$T$ of traces, we define
\begin{itemize}
    
    \item $T, \Pi, i \models \prop_\pi $ if  $\prop\in\Pi(\pi)(i)$, 
    
    \item $T, \Pi, i \models \propvar $ if for all $ t \in T$ we have $\propvar\in t(i)$,
    
    \item $T, \Pi, i \models \lnot \psi $ if  $T, \Pi, i \not\models \psi$, 
    
    \item $T, \Pi, i \models \psi_1 \lor \psi_2 $ if  $T, \Pi, i \models \psi_1$ or $T, \Pi, i \models \psi_2$, 
    
    \item $T, \Pi, i \models \X \psi $ if  $T, \Pi, i+1 \models \psi$, 
    
    \item $T, \Pi, i \models \F \psi $ if  there is an $i' \ge i$ with $T, \Pi, i' \models \psi$, 
    
    \item $T, \Pi, i \models \exists \pi.\, \phi $ if  there exists a $t \in T$ such that $T, \Pi[\pi\mapsto t], i \models \phi$, 
    
    \item $T, \Pi, i \models \forall \pi.\ \phi $ if $T, \Pi[\pi\mapsto t], i \models \phi$ for all $t \in T$, 
    
    \item $T, \Pi, i \models \exists \propvar.\ \phi $ if  there exists a $t_\propvar \in (\pow{\set{\propvar}})^\omega$ such that $T[\propvar\mapsto t_\propvar], \Pi, i \models \phi$, and  

    \item $T, \Pi, i \models \forall \propvar.\ \phi $ if $T[\propvar\mapsto t_\propvar], \Pi, i \models \phi$ for all $t_\propvar \in (\pow{\{\propvar\}})^\omega$. 
\end{itemize}

We say that a nonempty set~$T$ of traces satisfies a sentence~$\phi$, written $T \models\phi$, if $T, \Pi_\emptyset, 0 \models \phi$ where $\Pi_\emptyset$ is the trace assignment with empty domain. 
We then also say that $T$ is a model of $\phi$.
A transition system~$\tsys$ satisfies $\phi$, written $\tsys\models\phi$, if $\traces(\tsys) \models \phi$.

While it is known that $\hyqptl$ finite-state satisfiability is $\Sigma^0_1$-complete~\cite{Rabe16diss}, i.e., complete for the recursively-enumerable languages, and $\hyqptl$ model-checking is $\tower$-complete~\cite{Rabe16diss}, the exact complexity of $\hyqptl$ satisfiability is open: it is only known to be $\Sigma_1^1$-hard, as satisfiability is $\Sigma_1^1$-complete for its fragment $\hyltl$~\cite{hyperltlsat}.

\subsection{HyperQPTL Satisfiability}

In this subsection, we study the $\hyqptl$ satisfiability problem, which asks whether, for a given $\hyqptl$ sentence~$\phi$, there is a nonempty~$T$ such that $T \models \phi$.
We show that it is $\Sigma^2_1$-complete: intuitively, quantification over traces is equivalent to quantification over subsets of $\nats$: a trace~$t$ encodes the set~$\set{n \in\nats\mid \prop \in t(n)}$ and vice versa, where $\prop$ is some designated proposition.
Furthermore, existential quantification over sets of sets of natural numbers is equivalent to the selection of a set of traces at the meta-level of the satisfiability problem. 

For this to be correct, we need to enforce that the model of a $\hyqptl$ sentence contains enough traces to encode all subsets of $\nats$. 
As usual, we use $\contcard$ to denote the cardinality of the continuum, or equivalently, the cardinality of $(\pow{\ap})^\omega$ for each finite $\ap$ and the cardinality of $\pow{\nats}$.
As models of $\hyqptl$ sentences are sets of traces, $\contcard$ is a trivial upper bound on the size of models.
It is straightforward to show that this upper bound is tight, which also implies that there are sentences whose models allow us to encode all subsets of $\nats$.

\begin{theorem}
There is a satisfiable $\hyqptl$ sentence that has only models of cardinality~$\contcard$.
\end{theorem}

\begin{proof}
    Fix $\ap = \set{\inprop,\propvar}$ and let $\theta_\all= \forall \propvar.\ \exists \pi.\ \G(\propvar\leftrightarrow \inprop_\pi)$, which requires that the $\set{\inprop}$-projection of any model of $\phi_\all$ is equal to $(\pow{\set{\inprop}})^\omega$.
    As $(\pow{\set{\inprop}})^\omega$ has cardinality~$\contcard$, every model of $\theta_\all$ has cardinality~$\contcard$.
\end{proof}

This, and the fact that addition and multiplication can be \myquot{implemented} in $\hyltl$~\cite{hyperltlsat} suffice to prove our lower bound on $\hyqptl$ satisfiability.

\begin{lemma}
\label{lemma_satcomplexity_hyqptl_lb}
$\hyqptl$ satisfiability is $\Sigma^2_1$-hard.
\end{lemma}

\begin{proof}
We need to show that every language in $\Sigma^2_1$ is polynomial-time reducible to $\hyqptl$ satisfiability. To this end, we first present a polynomial-time computable function~$f$ mapping sentences of the form~$\phi = \exists \mathcal{Y}_1\ldots \exists \mathcal{Y}_k.\ \psi $, with third-order variables~$\mathcal{Y}_j$ and with $\psi$ only containing first-order and second-order quantifiers, to $\hyqptl$ sentences~$f(\phi)$ such that $\natsstruct \models \phi$ if and only if $f(\phi)$ is satisfiable.
To this end, we fix $\ap = \set{\inprop,\propvar, \marker_1, \ldots, \marker_k, \argone, \argtwo, \res, \add, \mult}$. We explain the roles of the different propositions along the way. 

Recall that the sentence~$\theta_\all$ ensures that the $\set{\inprop}$-projection of each of its models is equal to $(\pow{\set{\inprop}})^\omega$. 
Hence, by ignoring all propositions but $\inprop$, we can use traces to encode sets of natural numbers and natural numbers (as singleton sets) and each model contains the encoding of each set of natural numbers.
In our encoding, a trace bound to a trace variable~$\pi$ encodes a singleton set if and only if the formula~$(\neg \inprop_{\pi})\U(\inprop_{\pi} \wedge \X\G\neg \inprop_{\pi}) $ is satisfied.

Thus, a set of sets of natural numbers can be encoded by marking traces (encoding sets) by a marker signifying whether it is in the set (the marker holds at the first position) or not (the marker does not hold at the first position). 
Note that several traces encode the same set, hence we need to ensure that the marking is consistent among these traces. 
We use a distinct marker~$\marker_j$ for each $\mathcal{Y}_j$ appearing in $\phi$ and the formula 
\begin{multline*}
\theta_\consistent = \bigwedge\nolimits_{j=1}^k \forall\pi.\ \X\G\neg(\marker_j)_\pi \wedge \forall\pi'.\ \Big(
\big(\G(\inprop_{\pi} \leftrightarrow \inprop_{\pi'})
\big) 
\rightarrow 
\big((\marker_j)_{\pi} \leftrightarrow (\marker_j)_{\pi'} \big)
\Big)    
\end{multline*}
to express consistency: the first line expresses that only the first position of a trace may (but must not) be labeled with markers and the second line expresses that when two traces have the same truth value of $\inprop$ everywhere (i.e., they encode the same set), then they have the same marking. 

So, every model~$T$ of $\theta_\all \wedge \theta_\consistent$ encodes every subset of $\nats$ as a trace in $T$ and also $k$ sets of subsets of $\nats$ via the (consistent) marking of traces, i.e., the existential quantification over a model in the satisfiability problem simulates the existential third-order quantification of the $\mathcal{Y}_j$.

Finally, Fortin et al.\ showed that addition and multiplication can be \myquot{implemented} in $\hyltl$~\cite{hyperltlsat}: 
Let $T_\plustimes$ be the set of all traces $t \in (\pow{\ap})^\omega$ such that there are unique $n_1, n_2, n_3 \in \nats$  with $\argone \in t(n_1)$, $\argtwo \in t(n_2)$, and $\res \in t(n_3)$, and either 
    \begin{itemize}
        \item $\add \in t(n)$ and $\mult \notin t(n)$ for all $n$, and $n_1+n_2 = n_3$, or
        \item $\mult \in t(n)$ and $\add \notin t(n)$ for all $n$, and $n_1 \cdot n_2 = n_3$.
    \end{itemize}
There is a satisfiable $\hyltl$ sentence $\theta_\plustimes$ such that the $\set{\argone, \argtwo, \res, \add, \mult}$-projection of every model of $\theta_\plustimes$ is $T_\plustimes$~\cite[Theorem 5.5]{hyperltlsat}.
As $\hyltl$ is a fragment of $\hyqptl$, we can use $\theta_\plustimes$ to construct our desired formula.

Now, we define for $\phi = \exists \mathcal{Y}_1\ldots \exists \mathcal{Y}_k.\ \psi $
\[
f(\phi)= \theta_\all \wedge \theta_\consistent \wedge \theta_\plustimes \wedge \hyperize(\psi) 
\]
where $\hyperize(\psi)$ is defined inductively as follows:
\begin{itemize}

    \item For second-order variables~$Y$, $\hyperize(\exists Y.\ \psi) = \exists \pi_Y.\ \hyperize(\psi) $ and $\hyperize(\forall Y.\ \psi) = \forall \pi_Y.\ \hyperize(\psi) $.
    
    \item For first-order variables~$y$,
    $\hyperize(\exists y.\ \psi) = \exists \pi_y.\  ((\neg \inprop_{\pi_y})\U(\inprop_{\pi_y} \wedge \X\G\neg \inprop_{\pi_y}))  \land \hyperize(\psi) $ and 
    $\hyperize(\forall y.\ \psi) = \forall \pi_y.\ ((\neg \inprop_{\pi_y})\U(\inprop_{\pi_y} \wedge \X\G\neg \inprop_{\pi_y}))  \rightarrow \hyperize(\psi) $,
    
    \item $\hyperize(\psi_1 \lor \phi_2) = \hyperize(\psi_1) \lor \hyperize(\psi_2)$,
    
    \item $\hyperize(\lnot \psi) = \lnot \hyperize(\psi) $,
    
    \item For third-order variables~$\mathcal{Y}_j$ and second-order variables~$Y$,  
    $\hyperize(Y\in \mathcal{Y}_j) = (\marker_j)_{\pi_Y}$,

    \item For second-order variables~$Y$ and first-order variables~$y$,  
    $\hyperize(y\in Y) = \F(\inprop_{\pi_y} \land \inprop_{\pi_Y})$,
    
    \item For first-order variables~$y,y'$, $\hyperize(y<y') =$\newline$ \F(\inprop_{\pi_y} \land \X\F\inprop_{\pi_{y'}})$,
    
    \item For first-order variables~$y_1,y_2,y$, $\hyperize(y_1+y_2=y) = \exists \pi.\ \add_\pi \land \F(\inprop_{\pi_{y_1}}\land\argone_\pi) \land \F(\inprop_{\pi_{y_2}}\land\argtwo_\pi) \land \F(\inprop_{\pi_y}\land\res_\pi)$ and $\hyperize(y_1 \cdot y_2=y) = \exists \pi.\ \mult_\pi \land \F(\inprop_{\pi_{y_1}}\land\argone_\pi) \land \F(\inprop_{\pi_{y_2}}\land\argtwo_\pi) \land \F(\inprop_{\pi_y}\land\res_\pi)$.
\end{itemize}
While $f(\phi)$ is not necessarily in prenex normal form, it can easily be brought into prenex normal form, as there are no quantifiers under the scope of a temporal operator.

An induction shows that we have that $\natsstruct\models\phi$ if and only if $f(\phi)$ is satisfiable: Over the structure of $\psi$, one proves for all variable assignments~$\alpha$ of the free variables of a subformula~$\psi'$ that $\natsstruct, \alpha \models \psi'$ if and only if $T_\alpha, \Pi_\alpha, 0 \models \hyperize(\psi')$, where $T_\alpha$ is a set containing enough traces to satisfy $\theta_\all \wedge \theta_\plustimes$, which are additionally marked according to the values~$\alpha(\mathcal{Y}_j)$ such that $\theta_\consistent$ is satisfied, and where $\Pi_\alpha$ mimics the assignment~$\alpha$ to first- and second-order variables as explained above.
Then, the remaining existential third-order quantifiers are mimicked by the choice of $T$, as also explained above.
We leave the fully formal definition and the actual inductive proof, which follows from the remarks above, to the reader.

Now, consider a language~$L$ in $\Sigma^2_1$, i.e., it is of the form~$\set{ n \in\nats \mid \natsstruct\models\phi(n)}$ where $\phi(x)$ is a formula of the form~$\exists \mathcal{Y}_1\ldots \exists \mathcal{Y}_k.\ \psi(x,\mathcal{Y}_1, \ldots, \mathcal{Y}_k) $, with first-order variable~$x$ and third-order variables~$\mathcal{Y}_j$, and with $\psi$ only containing first-order and second-order quantifiers.

There is a function that, given $n\in\nats$ returns a first-order formula~$\theta_{=n}(x)$ with a single free first-order variable~$x$ that is only satisfied in $\natsstruct$ if $x$ is assigned $n$.
This function can be implemented in polynomial time in $\log n$.
So, we have $n \in L$ if and only if $f(\exists x.\ \theta_{=n}(x) \wedge \phi(x))$ is satisfiable.
Thus, we have completed the desired reduction: $\hyqptl$ satisfiability is indeed $\Sigma^2_1$-hard.
\end{proof}

Now, we consider the upper bound. Here, we use the fact that traces can be encoded as sets of natural numbers via an encoding that is \myquot{implementable} in arithmetic.
This allows us to encode the satisfiability problem in arithmetic.

\begin{lemma}
\label{lemma_satcomplexity_hyqptl_ub}
$\hyqptl$ satisfiability is in $\Sigma^2_1$.
\end{lemma}

\begin{proof}
Formulas~$\phi$ of $\hyqptl$ can be encoded via natural numbers~$\enc(\phi)$ in a way that the encoding is implementable in first-order arithmetic (e.g., by a Gödel numbering). 
Relying on such an encoding, we present a formula~$\theta(x)$ of arithmetic with a single free (first-order) variable~$x$ such that a $\hyqptl$ sentence~$\phi$ is satisfiable if and only if $\natsstruct \models \theta(\enc(\phi))$. 
To obtain our result, $\theta$ may only use existential third-order quantification, but arbitrary second- and first-order quantification.

Intuitively, we  express the existence of a nonempty set of traces and the existence of some structures that prove that this set is a model of $\phi$, generalizing the technique used to prove  $\Sigma_1^1$-membership of $\hyltl$ satisfiability~\cite{hyperltlsat}.
These structures include in particular Skolem functions for the existential variables. 
Then, we express in arithmetic that every trace assignment (which can be encoded by a type~$1$ object) that is consistent with the Skolem functions satisfies the maximal quantifier-free subformula of $\phi$.
To evaluate this subformula, we implement the semantics of quantifier-free $\hyqptl$ in arithmetic by existentially quantifying a function that assigns to each trace assignment~$\Pi$ and each quantifier-free subformula~$\psi'$ of $\phi$ the truth value of $\psi'$ with respect to $\Pi$. This function is again an existentially quantified type~$2$ object and its correctness can be expressed in first-order arithmetic.
In the following, we introduce the necessary encoding of traces, an equivalent semantics of $\hyqptl$ that is more convenient for the construction here, Skolem functions, and then the construction of the formula~$\theta$.

For the remainder of the proof, let us fix a $\hyqptl$ sentence~$\phi$ over $\ap$.
To simplify our constructions, we assume w.l.o.g.\ that each trace variable and each proposition is quantified at most once in $\phi$. Also, we assume w.l.o.g.\ that $\ap$ and the set of trace variables appearing in $\phi$ are disjoint finite subsets of $\nats$.
These properties can be guaranteed by renaming trace variables and propositions.

Let us first explain how we encode traces as sets of natural numbers. 
Let $\pair \colon \nats\times\nats \rightarrow\nats$ denote Cantor's pairing function defined as $\pair(i,j) = \frac{1}{2}(i+j)(i+j+1) +j$, which is a bijection and can be implemented in arithmetic.
Then, we encode a trace~$t \in (\pow{\ap})^\omega$ by the set~$S_t =\set{\pair(i,\prop) \mid i \in \nats \text{ and } \prop \in t(i)} \subseteq \nats$.
Now, one can write a formula~$\phi_\istrace(Y)$ which is satisfied in $\natsstruct$ if and only if the interpretation of $Y$ encodes a trace over $\ap$~\cite{DBLP:conf/csl/Frenkel025}.

Recall that our plan is to existentially quantify a model (a set of traces), which is encoded by a set of sets of natural numbers, i.e., a type~$2$ object. 
Then, we want to evaluate the formula over that model using arithmetic.
Furthermore, the semantics of $\hyqptl$ as defined in Subsection~\ref{subsec_hyperqptl_syse} updates the model~$T$ every time a proposition is quantified, i.e., from $T$ to $T[\propvar\mapsto t_\propvar]$ for some trace~$t_\propvar \in (\pow{\propvar})^\omega$.
To capture this update naively, in case $\propvar$ is universally quantified, requires universal quantification of a type~$2$ object that captures the encoding of $T[\propvar\mapsto t_\propvar]$.
However, this would not yield the desired $\Sigma_1^2$ upper bound.

However, we do not need to update the model~$T$ as long as we keep track of $t_\propvar$, as all traces in $T$ have the same truth values w.r.t.\ $\propvar$, i.e., those of $t_\propvar$.
Hence, we keep track of the truth values assigned to quantified propositions using an assignment mapping propositions~$\propvar$ to sequences in $(\pow{\set{\propvar}})^\omega$.
Thus, in this proof we work with two assignments, a trace assignment and a propositional assignment. 
To clearly distinguish them, we use the symbols~$\Pi_t$ and $\Pi_p$, respectively.

Next, we introduce the modified semantics. For a position~$i$ and $T$, $\Pi_t$, and $\Pi_p$ as above, we define 
\begin{itemize}

    \item $T, \Pi_t, \Pi_p, i \models \prop_\pi $ if  $\prop $ is in the domain of $\Pi_p$ and $\prop\in\Pi_p(\prop)(i)$, or if $\prop$ is not in the domain of $\Pi_p$ and $\prop \in \Pi_t(\pi)(i)$, 
    
    \item $T, \Pi_t, \Pi_p, i \models \propvar $ if $\propvar\in \Pi_p(\propvar)(i)$,
    
    \item $T, \Pi_t, \Pi_p, i \models \lnot \psi' $ if  $T, \Pi_t, \Pi_p, i \not\models \psi'$, 
    
    \item $T, \Pi_t, \Pi_p, i \models \psi_1' \lor \psi_2' $ if  $T, \Pi_t, \Pi_p, i \models \psi_1'$ or $T, \Pi_t, \Pi_p, i \models \psi_2'$, 
    
    \item $T, \Pi_t, \Pi_p, i \models \X \psi' $ if  $T, \Pi_t, \Pi_p, i+1 \models \psi'$,
    
    \item $T, \Pi_t, \Pi_p, i \models \F \psi' $ if  there is an $i' \ge i$ such that $T, \Pi_t, \Pi_p, i' \models \psi'$,

    \item $T, \Pi_t, \Pi_p,  i \models \exists \pi.\, \phi $ if  there exists a trace~$t \in T$ such that $T, \Pi_t[\pi\mapsto t], \Pi_p, i \models \phi$, 
    
    \item $T, \Pi_t, \Pi_p, i \models \forall \pi.\ \phi $ if for all traces~$t \in T$ we have $T, \Pi_t[\pi\mapsto t], \Pi_p, i \models \phi$, 
    
    \item $T, \Pi_t, \Pi_p, i \models \exists \propvar.\ \phi $ if  there exists a trace~$t_\propvar \in (\pow{\set{\propvar}})^\omega$ such that $T, \Pi_t, \Pi_p[\propvar\mapsto t_\propvar], i \models \phi$, and  

    \item $T, \Pi_t, \Pi_p, i \models \forall \propvar.\ \phi $ if for all traces~$t_\propvar \in (\pow{\{\propvar\}})^\omega$ we have $T, \Pi_t, \Pi_p[\propvar\mapsto t_\propvar], i \models \phi$. 
\end{itemize}
We have that $T, \Pi_{\emptyset}, 0 \models \phi$ (i.e., in the semantics as in Subsection~\ref{subsec_hyperqptl_syse}) if and only if $T, \Pi_{t,\emptyset}, \Pi_{p,\emptyset}, 0 \models \phi$ (i.e., in the new semantics) for all nonempty~$T$ and all $\hyqptl$ sentences~$\phi$, where $\Pi_{t,\emptyset}$ and $\Pi_{p,\emptyset}$ denote the assignments with empty domain.
Also, for quantifier-free formulas~$\psi'$, $T, \Pi_t, \Pi_p, i \models \psi'$ is independent of $T$, i.e., we have $T, \Pi_t, \Pi_p, i \models \psi'$ if and only if $T', \Pi_t, \Pi_p, i \models \psi'$ for all $T$ and $T'$.
This is due to the fact that $T$ is only referred to in the cases of trace quantification in the new semantics.
This is useful later when we construct $\theta$.

To simplify the construction of the formula~$\theta(x)$, we capture existential quantification of traces and propositions by Skolem functions.
This allows us replace all quantifiers by a single universal quantifier that ranges over assignments that are consistent with the Skolem functions. 

Fix some nonempty set~$T$ of traces over $\ap$.
Let $\pi$ be an existentially quantified trace variable in $\phi$ and let $\pi_1, \ldots, \pi_k$ be the trace variables that are universally quantified before $\pi$ and let $\propvar_1,\ldots, \propvar_\ell$ be the propositional variables that are universally quantified before $\pi$.
A $T$-Skolem function for $\pi$ is a function mapping a $k$-tuple of traces in $T$ and an $\ell$-tuple of traces (the $i$-th one over $\set{\propvar_i}$) to a trace in $T$.
Similarly, let $\propvar$ be an existentially quantified propositional variable in $\phi$ and let $\pi_1, \ldots, \pi_k$ be the trace variables that are universally quantified before $\pi$ and let $\propvar_1,\ldots, \propvar_\ell$ be the propositional variables that are universally quantified before $\pi$.
A $T$-Skolem function for $\propvar$ is a function mapping a $k$-tuple of traces in $T$ and an $\ell$-tuple of traces (the $i$-th one over $\set{\propvar_i}$) to a trace over $\set{\propvar}$.

Consider a pair~$(\Pi_t, \Pi_p)$ of a trace assignment and a propositional assignment such that all universally quantified trace variables in $\phi$ are mapped to some trace in $T$ by $\Pi_t$ and all universally quantified propositional variables~$\propvar$ are mapped to some trace in $(\pow{\set{\propvar}})^\omega$ by $\Pi_p$.
We say that $(\Pi_t, \Pi_p)$ is consistent with a $T$-Skolem function~$f$ for an existentially quantified trace variable~$\pi$ if $\Pi_t(\pi) = f( \Pi_t(\pi_1), \ldots, \Pi_t(\pi_k), \Pi_p(\propvar_1), \ldots, \Pi_p(\propvar_\ell) )$ and it is consistent with a $T$-Skolem function~$f$ for an existentially quantified propositional variable~$\propvar$ if $\Pi_t(\propvar) = f( \Pi_t(\pi_1), \ldots, \Pi_t(\pi_k), \Pi_p(\propvar_1), \ldots, \Pi_p(\propvar_\ell) )$.
In both cases, the $\pi_j$ and the $\propvar_j$ are the trace and propositional variables universally quantified before $\pi$ and $\propvar$, respectively. 

Recall that we have fixed a $\hyqptl$ sentence~$\phi$. 
Let  $\psi$ be the maximal quantifier-free subformula of $\phi$.
Now, we have $T \models \phi$ (in the semantics from Subsection~\ref{subsec_hyperqptl_syse}) if and only if there are $T$-Skolem functions for all existentially quantified trace variables and for all existentially quantified propositional variables such that 
\begin{itemize}
    \item for every trace assignment~$\Pi_t$ whose domain contains all trace variables of $\psi$, that maps universally quantified variables to $T$, and that is consistent with the Skolem functions and
    \item for every propositional assignment~$\Pi_p$ whose domain contains all propositional variables in $\psi$, that maps universally quantified~$\propvar$ to traces over $\set{\propvar}$, and that is consistent with the Skolem functions,
\end{itemize}
we have $\Pi_t, \Pi_p, 0 \models \psi$.
Thus, we have related the original semantics with the semantics using Skolem functions and propositional assignments.

Finally, we need to witness the satisfaction of a quantifier-free formula by assignments~$\Pi_t$ and $\Pi_p$. 
Let $\Psi$ be the set of subformulas of $\psi$, which is the maximal quantifier-free subformula of $\phi$.
Now, we define the expansion~$e_{\psi, \Pi_t, \Pi_p} \colon \Psi\times\nats \rightarrow \set{0,1}$ of $\psi$ with respect to a trace assignment~$\Pi_t$ and a propositional assignment~$\Pi_p$ via
\[
e_{\psi, \Pi_t, \Pi_p}(\psi', i)=
\begin{cases}
1 &\text{ if } \Pi_t, \Pi_p, i      \models \psi',\\
0 &\text{ if } \Pi_t, \Pi_p, i \not \models \psi'.\\
\end{cases}
\]
In fact, $e_{\psi, \Pi_t, \Pi_p}$ is completely characterized by the following consistency conditions that only depend on $\Pi_t$ and $\Pi_p$:
\begin{itemize}
    \item $e_{\psi, \Pi_t, \Pi_p}(\prop_\pi, i) = 1$ if and only if either $\prop$ is in the domain of $\Pi_p$ and $\prop \in \Pi_p(\prop)(i)$ or if $\prop$ is not in the domain of $\Pi_p$ and $\prop\in\Pi_t(\pi)(i)$,
    \item $e_{\psi, \Pi_t, \Pi_p}(\propvar, i) = 1$ if and only if $\propvar\in \Pi_p(\propvar)(i)$,
    \item $e_{\psi, \Pi_t, \Pi_p}(\neg\psi', i) = 1$ if and only if $e_{\psi, \Pi_t, \Pi_p}(\psi', i) = 0$,
    \item $e_{\psi, \Pi_t, \Pi_p}(\psi_1' \vee \psi_2', i) = 1$ if and only if $e_{\psi, \Pi_t, \Pi_p}(\psi_1', i) = 1$ or $e_{\psi, \Pi_t, \Pi_p}(\psi_2', i) = 1$, 
    \item $e_{\psi, \Pi_t, \Pi_p}(\X\psi', i) = 1$ if and only if $e_{\psi, \Pi_t, \Pi_p}(\psi', i+1) = 1$, and 
    \item $e_{\psi, \Pi_t, \Pi_p}(\F\psi', i) = 1$ if and only if $e_{\psi, \Pi_t, \Pi_p}(\psi', i') = 1$ for some $i' \ge i$.
\end{itemize}
Note that all these conditions are expressible using a formula of first-order arithmetic (that depends on $\psi$).

Now, the formula~$\theta$ expresses the existence of the following third-order objects:
\begin{itemize}
    \item A nonempty set~$\mathcal{T}$ of sets of natural numbers intended to encode a set of traces over $\ap$.
    \item A function $\mathcal{S} \colon \nats \times (\pow{\nats})^* \rightarrow \pow{\nats}$ to be interpreted as Skolem functions, i.e., it maps variable or proposition names (the first argument) and (encodings of) sequences of traces (the second argument) to traces.
    \item A function~$\mathcal{E} \colon (\pow{\nats})^* \times \nats \times \nats \rightarrow \set{0,1}$ to be interpreted as expansions as follows: if the first argument~$A$ encodes a trace assignment~$\Pi_t$ and a propositional assignment~$\Pi_p$ whose domains contain exactly the variables quantified in $\phi$, then $x,y \mapsto \mathcal{E}(A, x,y )$ encodes the expansion~$e_{\psi, \Pi_t, \Pi_p}$, i.e., $x$ encodes a subformula of $\psi$ and $y$ a position.
\end{itemize}
Note that $\mathcal{S}$ and $\mathcal{E}$ can also be encoded as sets of sets of natural numbers, using the fact that a sequence~$(S_1, \cdots, S_k) \in(\pow{\nats})^*$ can be encoded by the set~$\set{\pair(n,j) \mid n \in S_j} \in \pow{\nats}$, and that a function~$f\colon \nats \rightarrow \nats$ can be encoded by $\set{\pair(n, f(n)) \mid n \in \nats} \subseteq \nats$. 
All these encodings are implementable in arithmetic.

Then, we express the following properties using only first- and second-order quantification: 
\begin{itemize}
    \item $\mathcal{T}$ does indeed only contain encodings of traces over $\ap$, i.e., it encodes a set~$T$ of traces (using the formula~$\phi_\istrace$ from above).
    \item The image of~$\mathcal{S}$ only contains traces over $T$, if the first argument is a trace variable, and traces over~$\set{\propvar}$, if the first input is a propositional variable~$\propvar$ (using a formula~$\phi_\skolemformat(\mathcal{S},\enc(\phi))$ of second-order arithmetic).
    \item For every $A \in (\pow{\nats})^*$ encoding a trace assignment~$\Pi_t$ and a propositional assignment~$\Pi_p$ such that 
    \begin{itemize}
        \item $(\Pi_t, \Pi_p)$ is consistent with the Skolem functions encoded by $\mathcal{S}$,

        \item all universally quantified trace variables in $\phi$ are mapped by $\Pi_t$ to some trace encoding of a trace in $T$, and
        \item all universally quantified propositional variables~$\propvar$ in $\phi$ are mapped by $\Pi_p$ to an encoding of a trace over $\set{\propvar}$,
        
    \end{itemize}
        we require that
      $\mathcal{E}(A,\cdot,\cdot)$ satisfies the consistency conditions (i.e., $\mathcal{E}$ indeed encodes the expansion) and that we have $\mathcal{E}(A,n_0,0) = 1$, where $n_0$ is the natural number encoding $\psi$ (using a formula~$\phi_\skolemcorrect(\mathcal{S}, \enc(\phi))$ of second-order arithmetic).
\end{itemize}
Thus, $\theta$ has the form
\begin{multline*}
\exists \mathcal{T}.\ 
\exists \mathcal{S}.\ 
\exists \mathcal{E}.\ \exists Y.\ Y \in \mathcal{T} \wedge 
\forall Y.\ (Y \in \mathcal{T} \rightarrow \phi_\istrace(Y)) \wedge \phi_\skolemformat(\mathcal{S},\enc(\phi)) \wedge \phi_\skolemcorrect(\mathcal{S}, \enc(\phi)).    
\end{multline*}
We leave the tedious, but standard, construction of $\phi_\skolemformat(\mathcal{S},\enc(\phi))$ and $\phi_\skolemcorrect(\mathcal{S}, \enc(\phi))$ to the reader.

Putting all the pieces together yields that $\phi$ is satisfiable if and only if $\natsstruct \models \theta(\enc(\phi))$.
\end{proof}

Now, our main result on $\hyqptl$ satisfiability is a direct consequence of Lemma~\ref{lemma_satcomplexity_hyqptl_lb} and Lemma~\ref{lemma_satcomplexity_hyqptl_ub}.

\begin{theorem}
\label{theorem_satcomplexity_hyqptl}
$\hyqptl$ satisfiability is $\Sigma^2_1$-complete.
\end{theorem}

\section{\texorpdfstring{$\hyqptlplus$}{HyperQPTL\boldmath$^+$}}
\label{sec_hyqptlplus}

In this section, we introduce $\hyqptlplus$ and then settle the complexity of its satisfiability, finite-state satisfiability, and model-checking problems.

\subsection{Syntax and Semantics}

In $\hyqptl$, quantification over a proposition~$\propvar$ is interpreted as labeling each trace by the same sequence~$t_\propvar$ of truth values for $\propvar$, i.e., the assignment of truth values is uniform. 
However, one can also consider a non-uniform labeling by truth values for $\propvar$. This results in the logic~$\hyqptlplus$.

The syntax of $\hyqptlplus$~\cite{FinkbeinerHHT20} is very similar to that of $\hyqptl$, one just drops the atomic formulas of the form~$\propvar$, i.e., atomic propositions that are not labeled by trace variables:
\begin{align*}
    \phi &{} \cceq {} \exists \pi.\ \phi \mid \forall \pi.\ \phi \mid \exists \propvar.\ \phi \mid \forall \propvar.\ \phi \mid \psi \\
    \psi&{} \cceq{} \prop_\pi \mid \lnot \psi \mid \psi \lor \psi \mid \X \psi \mid \F \psi
\end{align*}
Here $\prop$ and $\propvar$ range over $\ap$ and $\pi$ ranges over $\var$.
The semantics are also similar, we just change the definition of propositional quantification as follows:
\begin{itemize}
    \item $T, \Pi, i \models \exists \propvar.\ \phi $ if  there exists a $T' \subseteq (\pow{\ap})^\omega$ such that $T =_{\ap\setminus\set{\propvar}} T'$ and $T', \Pi, i  \models \phi$, and  

    \item $T, \Pi, i \models \forall \propvar.\ \phi $ if for all $T' \subseteq (\pow{\ap})^\omega$ such that $T =_{\ap\setminus\set{\propvar}} T'$  we have $T', \Pi, i \models \phi$. 
\end{itemize}

It is known that model-checking $\hyqptlplus$ is undecidable~\cite{FinkbeinerHHT20}, but its exact complexity is open, as is the complexity of satisfiability and finite-state satisfiability.

In the following, we show that $\hyqptlplus$ is equally expressive as $\sohyltl$, which allows us to transfer the complexity results for $\sohyltl$ to $\hyqptlplus$.

\subsection{Second-order \texorpdfstring{$\hyltl$}{HyperLTL}}
\label{subsec_sohyperltl}

We continue by introducing the syntax and semantics of $\sohyltl$~\cite{DBLP:conf/cav/BeutnerFFM23}.
Let $\fovar$ be a set of first-order trace variables (i.e., ranging over traces) and $\sovar$ be a set of second-order trace variables (i.e., ranging over sets of traces) such that $\fovar \cap \sovar = \emptyset$.

The formulas of $\sohyltl$ are given by the grammar
\begin{align*}
\phi & {} \cceq {} \exists X.\ \phi \mid \forall X.\ \phi \mid \exists \pi \in X.\ \phi \mid \forall \pi \in X.\ \phi \mid \psi \\
\psi &{}  \cceq {} \prop_\pi \mid \neg \psi \mid \psi \vee \psi \mid \X \psi \mid \psi \U \psi    
\end{align*}
where $\prop$ ranges over $\ap$, $\pi$ ranges over $\fovar$, $X$ ranges over $\sovar$, and $\X$ (next) and $\U$ (until) are temporal operators. Conjunction, implication, equivalence, and the temporal operators eventually and always are defined as usual. 

The semantics of $\sohyltl$ is defined with respect to a variable assignment, i.e., a partial mapping~$\Pi \colon \fovar \cup \sovar \rightarrow (\pow{\ap})^\omega \cup \pow{(\pow{\ap})^\omega}$
such that
\begin{itemize}
    \item if $\Pi(\pi)$ for $\pi\in\fovar$ is defined, then $\Pi(\pi) \in (\pow{\ap})^\omega$ and
    \item if $\Pi(X)$ for $X \in\sovar$ is defined, then $\Pi(X) \in \pow{(\pow{\ap})^\omega}$.
\end{itemize}
Given a variable assignment~$\Pi$, a variable~$\pi \in \fovar$, and a trace~$t$, we denote by $\Pi[\pi \mapsto t]$ the assignment that coincides with $\Pi$ on all variables but $\pi$, which is mapped to $t$. 
Similarly, for a variable~$X \in \sovar$, and a set~$T$ of traces, $\Pi[X \mapsto T]$ is the assignment that coincides with $\Pi$ everywhere but $X$, which is mapped to $T$.

For a variable assignment~$\Pi$ and a position~$i$ we define
\begin{itemize}
	\item $\Pi, i \models \prop_\pi$ if $\prop \in \Pi(\pi)(i)$,
	\item $\Pi, i \models \neg \psi$ if $\Pi, i \not\models \psi$,
	\item $\Pi, i \models \psi_1 \vee \psi_2 $ if $\Pi, i \models \psi_1$ or $\Pi, i \models \psi_2$,
	\item $\Pi, i \models \X \psi$ if $\Pi, i+1 \models \psi$,
	\item $\Pi, i \models \psi_1 \U \psi_2$ if there is an $i' \ge i$ such that $\Pi, i' \models \psi_2$ and for all $i \le i'' < i'$ we have $\Pi, i'' \models \psi_1$,
	\item $\Pi, i \models \exists \pi \in X.\ \phi$ if there exists a trace~$t \in \Pi(X)$ such that $\Pi[\pi \mapsto t], i \models \phi$,
	\item $\Pi, i \models \forall \pi \in X.\ \phi$ if for all traces~$t \in \Pi(X)$ we have $\Pi[\pi \mapsto t], i \models \phi$,
    \item $\Pi, i \models \exists X.\ \phi$ if there exists a set~$T \subseteq (\pow{\ap})^\omega$ such that $\Pi[X\mapsto T], i \models \phi$, and
    \item $\Pi, i \models \forall X.\ \phi$ if for all sets~$T \subseteq (\pow{\ap})^\omega$ we have $\Pi[X\mapsto T], i \models \phi$.
\end{itemize}
Note that these are the standard semantics where second-order quantification ranges over arbitrary sets~\cite{DBLP:conf/cav/BeutnerFFM23}, not the closed-world semantics where second-order quantification only ranges over subsets of the model~\cite{DBLP:conf/csl/Frenkel025}.

We assume the existence of two distinguished second-order variables~$\univar, \unidisvar  \in \sovar$ such that $\univar$ refers to the set~$(\pow{\ap})^\omega$ of all traces, and $\unidisvar$ refers to the universe of discourse (the set of traces the formula is evaluated over).
A sentence is a formula in which only the variables $\univar,\unidisvar$ can be free. 
Thus, technically, a sentence does have free variables. Alternatively, one could treat the distinguished variables with distinguished syntax so that they are not free. 
However, for the sake of definitional simplicity, we refrain from doing so. 

We say that a nonempty set~$T$ of traces satisfies a $\sohyltl$ sentence~$\phi$, written $T \models \phi$, if $\Pi_\emptyset[\univar \mapsto (\pow{\ap})^\omega, \unidisvar \mapsto T],0\models \phi$, where $\Pi_\emptyset$ denotes the variable assignment with empty domain. 
In this case, we say that $T$ is a model of $\varphi$.
A transition system~$\tsys$ satisfies $\phi$, written $\tsys \models \phi$, if $\traces(\tsys)\models \phi$. 

\subsection{\texorpdfstring{$\hyqptlplus$}{HyperQPTL\boldmath$^+$} ``is'' \texorpdfstring{$\sohyltl$}{Second-order HyperLTL}}

In this subsection, we show that $\hyqptlplus$ and $\sohyltl$ are equally expressive by translating $\sohyltl$ into $\hyqptlplus$ and vice versa.

\begin{lemma}
\label{lemma_sohyltl2qptlplus}
There is a polynomial-time computable function~$f$ mapping $\sohyltl$ sentences~$\phi$ to $\hyqptlplus$ sentences~$f(\phi)$ such that we have $T \models \phi$ if and only if $T \models f(\phi)$ for all nonempty $T \subseteq (\pow{\ap})^\omega$.
\end{lemma}

\begin{proof}
Let $\phi$ be a $\sohyltl$ sentence.
We assume w.l.o.g.\ that each (trace and set) variable is quantified at most once in $\phi$. Further, we require that each set variable quantified in $\phi$ is different from $\unidisvar$ and $\univar$. 
These properties can always be achieved by renaming variables.

Our goal is to mimic quantification over sets of traces in $\hyqptlplus$ via (non-uniform) quantification over propositions.
Let ${\prop_1, \ldots, \prop_n}$ be the propositions appearing in $\phi$. Note that the truth values of other propositions can be ignored, as satisfaction of $\phi$ does not depend on them.
We existentially quantify fresh propositions~$\prop_1^\alltraces, \ldots ,\prop_n^\alltraces$ and then express in $\hyqptlplus$ that all traces over these fresh propositions are available to mimic set quantification. 
Then, we use further propositional quantification in $\hyqptlplus$ to mimic quantification over sets~$X$ of traces by labeling the traces over $\prop_1^\alltraces, \ldots ,\prop_n^\alltraces$ with a bit indicating whether they are in $X$ or not.
All other quantifiers, connectives, and temporal operators of $\sohyltl$ can directly be mimicked in $\hyqptlplus$.

In the following, we construct $\theta_\complete$ and $f'(\phi)$ such that 
\[f(\phi) = \exists \prop_1^\alltraces\ldots \exists \prop_n^\alltraces.\ (\theta_\complete \wedge f'(\phi)) \] 
which is satisfied by a nonempty set~$T \subseteq (\pow{\ap})^\omega$ of traces if and only if there is a set~$T'\subseteq (\pow{\ap})^\omega$ of traces such that  $T' \models \theta_\complete \wedge f'(\phi)$ and $T =_{\ap\setminus \set{\prop_1^\alltraces, \ldots ,\prop_n^\alltraces} } T'$.

Due to the second property, each trace in $T'$ has the form~$t\merge t'$ where $t$ is a trace from the $\ap \setminus \set{\prop_1^\alltraces, \ldots ,\prop_n^\alltraces}$-projection of $T$ and $t'$ is a trace over $\set{\prop_1^\alltraces, \ldots ,\prop_n^\alltraces}$.
Also, recall that the propositions~$\prop_j^\alltraces$ do not appear in $\phi$.
Thus, quantification over $T'$ mimics both quantification of traces over the propositions~$\prop_j$ relevant for evaluating $\phi$, and quantification over $(\pow{\set{\prop_1^\alltraces, \ldots ,\prop_n^\alltraces}})^\omega$.

As a first step, we need to express in $\hyqptlplus$ that the $\set{\prop_1^\alltraces, \ldots ,\prop_n^\alltraces}$-projection of a set of traces contains all traces over $\set{\prop_1^\alltraces, \ldots ,\prop_n^\alltraces}$ using the formula
\[
\theta_\complete = \forall \prop_1^\temp\!\!\!\ldots \forall \prop_n^\temp.\ \forall \pi.\ \exists \pi'.\! \bigwedge\nolimits_{j=1}^n \!\!\G((\prop_j^\temp)_\pi \!\leftrightarrow\! (\prop_j^\alltraces)_{\pi'})
\]
with another set~$\set{\prop_1^\temp, \ldots ,\prop_n^\temp}$ of fresh propositions used only in this subformula.
Intuitively, it expresses that for each trace over $\set{\prop_1^\temp, \ldots ,\prop_n^\temp}$ the trace obtained by replacing each~$\prop_j^\temp$ by $\prop_j^\alltraces$ is in each model of $\theta_\complete$, i.e., the $\set{\prop_1^\alltraces, \ldots ,\prop_n^\alltraces}$-projection of each model must indeed be the set of all traces over $\set{\prop_1^\alltraces, \ldots ,\prop_n^\alltraces}$.

Now, let us explain how we mimic the quantification over subsets of traces in $\hyqptlplus$: Similarly to the construction for the $\Sigma_1^2$-lower bound for $\hyqptl$ satisfiability (see Lemma~\ref{lemma_satcomplexity_hyqptl_lb}), we use propositional quantification over another fresh proposition to mark each trace in the $\set{\prop_1^\alltraces, \ldots ,\prop_n^\alltraces}$-projection of a model (say at the first position) to indicate whether it is in the set or not. 
If $X_1,\ldots, X_k$ are the second-order variables quantified in $\phi$, then we employ a marker~$\marker_j$ for each $j \in \set{1,\ldots, k}$.

However, we have to ensure that there are no two different traces that have the same $\set{\prop_1^\alltraces, \ldots ,\prop_n^\alltraces}$-projection, but different values of $\marker_j$. 
This can happen when quantifying $\marker_j$, as we do not mark the $\set{\prop_1^\alltraces, \ldots ,\prop_n^\alltraces}$-projection of the model, but the model itself.
To rule this out, we require that the marking is consistent using the formula
\begin{multline*}
\theta_\consistent^j = \forall \pi.\ (\X\G\neg(\marker_j)_\pi) \wedge \forall \pi'. \bigwedge\nolimits_{\ell = 1}^n \!\Big(\!\G\big( (\prop_\ell^\alltraces)_{\pi} \leftrightarrow (\prop_\ell^\alltraces)_{\pi'} \big)\! \Big) \rightarrow \Big(\!(\marker_j)_{\pi} \leftrightarrow (\marker_j)_{\pi'}\Big).    
\end{multline*}
It expresses that only the first position can be marked (but must not) and that if two traces have the same $\set{\prop_1^\alltraces, \ldots ,\prop_n^\alltraces}$-projection they have the same marking.

Note that traces in $T$ are still encoded by the \myquot{original} propositions~${\prop_1, \ldots ,\prop_n}$ while set quantification is mimicked by traces over the fresh propositions~${\prop_1^\alltraces, \ldots ,\prop_n^\alltraces}$. 
Hence, when mimicking the binding of a trace from some~$X_j$ or the distinguished variable~$\univar$ to some $\pi$, we need to ensure that we use the fresh propositions and not the original ones.
However, when mimicking the binding of a trace from $\unidisvar$, we need to use the original propositions.

We do so by replacing each $(\prop_j)_\pi$ by $(\prop_j^\alltraces)_\pi$, unless when binding a trace from $\unidisvar$ (which ranges over the set~$T$ we are evaluating $\phi$ over).
Thus, we define $f'$ as
\begin{itemize}
\item $f'(\exists X_j.\ \psi) = \exists \marker_j.\ \theta_\consistent^j \wedge f'(\psi)$, 
\item $f'(\forall X_j.\ \psi) = \forall \marker_j.\ \theta_\consistent^j \rightarrow f'(\psi)$, 
\item $f'(\exists \pi \in X_j.\ \psi) = \exists \pi.\ (\marker_j)_\pi \wedge f'(\replace{\pi}{\psi})$, where $\replace{\pi}{\psi}$ is the formula obtained from $\psi$ by replacing each subformula~$(\prop_j)_\pi$ by $(\prop_j^\alltraces)_\pi$ (note that we only replace propositions labeled by $\pi$, the variable quantified here),
\item $f'(\forall \pi \in X_j.\ \psi) = \forall \pi.\ (\marker_j)_\pi \rightarrow f'(\replace{\pi}{\psi})$, 
\item $f'(\exists \pi \in \univar.\ \psi) = \exists \pi.\  f'(\replace{\pi}{\psi})$,
\item $f'(\forall \pi \in \univar.\ \psi) = \forall \pi.\  f'(\replace{\pi}{\psi})$,
\item $f'(\exists \pi \in \unidisvar.\ \psi) = \exists \pi.\ f'(\psi)$, 
\item $f'(\forall \pi \in \unidisvar.\ \psi) = \forall \pi.\ f'(\psi)$, 
\item $f'(\neg \psi) = \neg f'(\psi)$, 
\item $f'(\psi_1 \vee \psi_2) = f'(\psi_1) \vee f'(\psi_2)$, 
\item $f'(\X \psi ) = \X f'(\psi)$,
\item $f'(\F \psi ) = \F f'(\psi)$, and
\item $f'(\prop_\pi) = \prop_\pi$.
\end{itemize}
While $f'(\phi)$ is not necessarily in prenex normal form, it can easily be brought into prenex normal form, as there are no quantifiers under the scope of a temporal operator.

An induction shows that we indeed have $T \models \phi$ if and only if $T \models f(\phi)$ for nonempty~$T$: 
Over the structure of $\phi$, one proves that for all nonempty sets~$T$ of traces and for all variable assignments~$\Pi$ of the free variables of a subformula~$\psi$ that $T, \Pi, i \models \psi$ if and only if $T', \Pi', i \models f'(\psi)$, where $T'$ is obtained from $T$ by completely \myquot{overwriting} the propositions~$\prop_j^\alltraces$ and assigning the marker propositions consistently according to $\Pi$, and where $\Pi'$ is the restriction of $\Pi$ to trace variables.
Again, we leave the tedious, but straightforward, details to the reader.
 \end{proof}

Now, we consider the other direction.

\begin{lemma}
\label{lemma_qptlplus2sohyltl}
There is a polynomial-time computable function~$f$ mapping $\hyqptlplus$ sentences~$\phi$ to $\sohyltl$ sentences~$f(\phi)$ such that we have $T \models \phi$ if and only if $T \models f(\phi)$ for all nonempty $T \subseteq (\pow{\ap})^\omega$.
\end{lemma}

\begin{proof}
The semantics of $\hyqptlplus$ can directly be expressed in $\sohyltl$. 
Intuitively, propositional quantification in $\hyqptlplus$ \myquot{updates} the set of traces that the trace quantifiers range over (i.e., the $T$ in $T, \Pi, i\models \psi$).
All other quantifiers, temporal operators, and connectives in $\hyqptlplus$ are also available in $\sohyltl$.

To mimic propositional quantification in $\sohyltl$, we use a dedicated second-order variable that explicitly stores the set of traces that the trace quantifiers in the $\hyqptlplus$ formula range over.
This set is initially equal to the set of traces the formula is evaluated over (i.e., the set~$\Pi(\unidisvar)$ in the setting of $\sohyltl$), and is updated with each quantification over a proposition~$\propvar$. 
As this update from $T$ to $T'$ has to satisfy $T =_{\ap\setminus\set{\propvar}} T'$, we need two set variables~$X_0$ and $X_1$, one for the old value and one for the new value, to be able to compare these two. 
As the old value is no longer needed after the update, we can then reuse the variable.
To keep track of which of the two variables~$X_0$ and $X_1$ is currently used to store the set of traces we are evaluating trace quantifiers over, we use a flag~$b \in \set{0,1}$, i.e., $X_b$ is the set variable storing the current value and $X_{1-b}$ is the variable storing the old value.
To enforce a correct update, we define the formula~$\theta_\propvar(X_b, X_{1-b})$ as
\begin{align*}
{}&{}\forall \pi \in X_{b}.\ \exists \pi' \in X_{1-b}. \bigwedge\nolimits_{\prop \in \ap\setminus\set{\propvar}} \G (\prop_\pi \leftrightarrow \prop_{\pi'}) \wedge \\
{}&{} \forall \pi \in X_{1-b}.\ \exists \pi' \in X_b. \bigwedge\nolimits_{\prop \in \ap\setminus\set{\propvar}} \G (\prop_\pi \leftrightarrow \prop_{\pi'})
,    
\end{align*}
which is satisfied by a variable assignment~$\Pi$ if and only if $\Pi(X_b) =_{\ap\setminus\set{\propvar}} \Pi(X_{1-b})$.

Now, we express the semantics of $\hyqptlplus$ in $\sohyltl$. 
Here, we use $X_0 = \unidisvar$ and let $X_1$ be some second-order variable other than $\unidisvar$.
Then, we define
\begin{itemize}
    \item $f'(\exists \pi.\ \psi,b) = \exists \pi\in X_b.\ f'(\psi,b)$,
    \item $f'(\forall \pi.\ \psi,b) = \forall \pi\in X_b.\ f'(\psi,b)$,
    \item $f'(\exists \propvar.\ \psi, b) = \exists X_{1-b}.\ \theta_\propvar(X_{1-b}, X_b) \wedge f'(\psi, 1-b)$,
    \item $f'(\forall \propvar.\ \psi, b) = \forall X_{1-b}.\ \theta_\propvar(X_{1-b}, X_b) \rightarrow f'(\psi, 1-b)$,
    \item $f'(\neg \psi, b)= \neg f'(\psi,b)$, 
    \item $f'(\psi_1 \vee \psi_2, b) = f'(\psi_1, b) \vee f'(\psi_2, b)$,
    \item $f'(\X\psi,b) = \X f'(\psi,b)$,
    \item $f'(\F\psi,b) = \F f'(\psi,b)$, and
    \item $f'(\prop_\pi, b) = \prop_\pi$.
\end{itemize}
Now, define $f(\phi) = f'(\phi,0)$.
While $f(\phi)$ is not necessarily in prenex normal form, it can easily be brought into prenex normal form, as there are no quantifiers under the scope of a temporal operator.

An induction shows that we indeed have $T \models \phi$ if and only if $T \models f(\phi)$: One shows that for all sets~$T$ of traces and all variable valuations~$\Pi$ of the free variables of a subformula~$\psi$ that $T, \Pi, i \models \psi$ if and only if $T, \Pi', i \models f(\psi, b)$, where $\Pi'$ is obtained from $\Pi$ by extending it with the appropriate assignment for the variable~$X_b$, which must satisfy $T =_{\ap \setminus \set{\propvar_1, \ldots, \propvar_k}} \Pi(X_b)$, where the $\propvar_j$ are the free propositional variables of $\psi$.
One last time, we leave the tedious, but straightforward, details to the reader.
\end{proof}

As $\sohyltl$ satisfiability, finite-state satisfiability, and model-checking are equivalent to truth in third-order arithmetic~\cite{DBLP:conf/csl/Frenkel025}, the translations presented in Lemma~\ref{lemma_sohyltl2qptlplus} and Lemma~\ref{lemma_qptlplus2sohyltl} imply that the same is true for $\hyqptlplus$.

\begin{theorem}
$\hyqptlplus$ satisfiability, finite-state satisfiability, and model-checking are equivalent to truth in third-order arithmetic.
\end{theorem}

Let us remark that two second-order variables suffice to translate $\hyqptlplus$ into $\sohyltl$.
Together with the converse translation presented in Lemma~\ref{lemma_sohyltl2qptlplus}, we conclude that every $\sohyltl$ sentence is equivalent to one with only two second-order variables.

\section{Conclusion}

We settled the exact complexity of the most important verification problems for $\hyqptl$ and $\hyqptlplus$.
For $\hyqptl$, we proved that satisfiability is $\Sigma^2_1$-complete while for $\hyqptlplus$, we proved that satisfiability, finite-state satisfiability, and model-checking are equivalent to truth in third-order arithmetic.
The latter results were obtained by showing that $\hyqptlplus$ and second-order $\hyltl$ have the same expressiveness.

\textbf{Acknowledgments.} Gaëtan Regaud was supported by the European Union.
Martin Zimmermann was supported by DIREC – Digital Research Centre Denmark. We thank the reviewers for their detailed comments, which improved the article considerably.

\bibliographystyle{plain}
\bibliography{bib}

\end{document}